%% file: main.tex
\pdfoutput=1
\documentclass[a4paper,UKenglish,cleveref, autoref, thm-restate,authorcolumns]{lipics-v2019}

\usepackage{enumitem}
\usepackage[numbers,sort]{natbib}

\title{Faster Dynamic Range Mode}

\author{Bryce Sandlund}{Cheriton School of Computer Science, University of Waterloo}{bcsandlund@gmail.com}{}{}

\author{Yinzhan Xu}{CSAIL, Massachusetts Institute of Technology}{xyzhan@mit.edu}{}{} 

\authorrunning{B. Sandlund and Y. Xu} 

\Copyright{Bryce Sandlund and Yinzhan Xu} 

\ccsdesc[500]{Theory of computation~Data structures design and analysis} 

\keywords{Range Mode, Min-Plus Product} 

\category{} 

\relatedversion{} 

\supplement{}

\acknowledgements{We would like to thank Virginia Vassilevska Williams for her valuable comments on a draft of this paper. We would like to thank the anonymous reviewers for their helpful suggestions.}

\nolinenumbers 

\hideLIPIcs  

\EventEditors{John Q. Open and Joan R. Access}
\EventNoEds{2}
\EventLongTitle{42nd Conference on Very Important Topics (CVIT 2016)}
\EventShortTitle{CVIT 2016}
\EventAcronym{CVIT}
\EventYear{2016}
\EventDate{December 24--27, 2016}
\EventLocation{Little Whinging, United Kingdom}
\EventLogo{}
\SeriesVolume{42}
\ArticleNo{23}

\theoremstyle{plain}
\newtheorem{problem}[theorem]{Problem}
\newtheorem{fact}[theorem]{Fact}

\makeatletter
\let\c@fconjecture\c@conjecture
\makeatother

\makeatletter
\let\c@fconj\c@conj
\makeatother

\def \qed {\hfill $\Box$}

\newcommand{\ignore}[1]{}

\def \polylog { \text{\rm polylog~} }

\DeclareMathOperator*{\argmin}{arg\,min}

\begin{document}

\maketitle

\begin{abstract}
In the dynamic range mode problem, we are given a sequence $a$ of length bounded by $N$ and asked to support element insertion, deletion, and queries for the most frequent element of a contiguous subsequence of $a$. In this work, we devise a deterministic data structure that handles each operation 
in worst-case $\tilde{O}(N^{0.655994})$ time, thus breaking the $O(N^{2/3})$ per-operation time barrier for this problem. The data structure is achieved by combining the ideas in Williams and Xu (SODA 2020) for batch range mode with a novel data structure variant of the Min-Plus product.


\end{abstract}

\newpage{}

\section{Introduction}
\label{sec:intro}
\input{intro}

\section{Related Work}
\label{sec:related}
\input{related}

\section{Preliminaries}
\label{sec:prelim}
\input{prelim}

\section{Main Algorithm}
\input{algo}

\bibliography{ref}

\end{document}

%% file: intro.tex

Given a sequence of elements $a_1, a_2, \ldots, a_n$, the dynamic range mode problem asks to support queries for the most frequent element in a specified subsequence $a_l, a_{l+1}, \ldots, a_r$ while also supporting insertion or deletion of an element at a given index $i$. 
The mode of a sequence of elements is one of the most basic data statistics, along with the median and the mean.
It is frequently computed in data mining, information retrieval, and data analytics.


The range mode problem seeks to answer multiple queries on distinct intervals of the data sequence without having to recompute each answer from scratch. Its study in the data structure community has shown that the mode is a much more challenging data statistic to maintain than other natural range queries: while range sum, min or max, median, and majority all support linear space dynamic data structures with poly-logarithmic or better time per operation~\cite{mih10,chan17,HMN11,gag17,elm11}, the current fastest dynamic range mode data structure prior to this paper requires a stubborn $O(n^{2/3})$ time per operation~\cite{zhms2018}. Indeed, range mode is one of few remaining classical range queries to which our currently known algorithms may be far from optimal. As originally stated by Brodal et al.~\cite{BGJS11} and mentioned by Chan et al.~\cite{C14} in 2011 and 2014, respectively, ``The problem of finding the most frequent element within a given array range is still rather open.''

The current best conditional lower bound, by Chan et al.~\cite{C14}, reduces multiplication of two $\sqrt{n} \times \sqrt{n}$ boolean matrices to $n$ range mode queries on a fixed array of size $O(n)$. This indicates that
if the current algorithm for boolean matrix multiplication is optimal, then
answering $n$ range mode queries on an array of size $O(n)$ cannot be performed in time better than $O(n^{3/2-\epsilon})$ time for $\epsilon > 0$ with combinatorial techniques, or $O(n^{\omega/2 - \epsilon})$ time for $\epsilon > 0$ in general, where $\omega < 2.373$ \cite{vmmult, legallmult} is the square matrix multiplication exponent. This reduction can be strengthened for dynamic range mode by reducing from the online matrix-vector multiplication problem~\cite{Henzinger15}. 
Using $O(n)$ dynamic range mode operations on a sequence of length $O(n)$, we can multiply a $\sqrt{n} \times \sqrt{n}$ boolean matrix with $\sqrt{n}$ boolean vectors given one at a time.
This indicates that a dynamic range mode data structure taking $O(n^{1/2 - \epsilon})$ time per operation for $\epsilon > 0$ is not possible with current knowledge.

Previous attempts indicate the higher $O(n^{2/3})$ per operation cost as the bound to beat~\cite{C14,zhms2018}. Indeed, $\tilde{O}(n^{2/3})$ time per operation\footnote{We use the $\tilde{O}(\cdot)$ notation to hide poly-logarithmic factors.} can be achieved with a variety of techniques, but crossing the $O(n^{2/3})$ barrier appears much harder.


Progress towards this goal has been established with the recent work of Williams and Xu~\cite{Williams20}. They show that by appealing to Min-Plus product of structured matrices, $n$ range mode queries on an array of size $n$ can be answered in $\tilde{O}(n^{1.4854})$ time, thus beating the combinatorial lower bound for batch range mode. This result also shows a separation between batch range mode and dynamic range mode: while batch range mode can be completed in $O(n^{1/2-\epsilon})$ time per operation, such a result for dynamic range mode would imply a breakthrough in the online matrix-vector multiplication problem.

Range mode is not the first problem  shown to be closely related to the Min-Plus product problem. It is well-known that the all-pairs shortest paths (APSP) problem is asymptotically equivalent to Min-Plus product~\cite{Fischer71}, in the sense that a $T(n)$ time algorithm to compute the Min-Plus product of two $n \times n$ matrices implies an $O(T(n))$ time algorithm for APSP in $n$-node graphs and vice versa. Although it is not known how to perform Min-Plus product of two arbitrary $n \times n$ matrices in time $O(n^{3-\epsilon})$ for $\epsilon > 0$, several problems reduce to Min-Plus products of matrices $A$ and $B$ which have nice structures that can be exploited. 
The simplest examples result by restricting edge weights in APSP problems~\cite{Seidel95,Shoshan99,Zwick02,Chan10,Yuster09}. Bringmann et al.~\cite{Bringmann17} show Language Edit Distance, RNA-folding, and Optimum Stack Generation can be reduced to Min-Plus product where matrix $A$ has small difference between adjacent entries in each row and column. Finally, the recent work of Williams and Xu~\cite{Williams20} reduces APSP in certain geometric graphs, batch range mode, and the maximum subarray problem with entries bounded by $O(n^{0.62})$ to a more general structured Min-Plus product, extending the result of Bringmann et al. 
All of the above structured Min-Plus products are solvable in truly subcubic $O(n^{3-\epsilon})$ time for $\epsilon > 0$, improving algorithms in the problems reduced to said product.


The connection and upper bound established by Williams and Xu~\cite{Williams20} of batch range mode to Min-Plus product suggest other versions of the range mode problem may be amenable to similar improvements. In particular, the ability to efficiently compute a batch of range mode queries via reducing to a structured Min-Plus product suggests that one might be able to improve the update time of dynamic range mode in a similar way.

\subsection{Our Results}
\label{ourresults}

In this paper, we break the $O(n^{2/3})$ time per operation barrier for dynamic range mode. We do so by adapting the result of Williams and Xu~\cite{Williams20}. Specifically, we define the following new type of data structure problem on the Min-Plus product that can be  applied to dynamic range mode, which may be of independent interest. Then we combine this data structure problem with the algorithm of Williams and Xu.

\begin{problem}[\textit{Min-Plus-Query} problem]
\label{mpq}
    During initialization, we are given two matrices $A, B$. For each query, we are given three parameters $i, j, S$, where $i, j$ are two integers, and $S$ is a set of integers. The query asks $\min_{k \not \in S} \{A_{i,k}+B_{k,j}\}$.
\end{problem}

Our performance theorem is the following.

\begin{theorem}
\label{maintheorem}
There exists a deterministic data structure for dynamic range mode on a sequence $a_1, \ldots, a_n$ that supports query, insertion, and deletion in worst-case $\tilde{O}(N^{0.655994})$ time per operation, where $N$ is the maximum size of the sequence at any point in time.
The space complexity of the data structure is $\tilde{O}(N^{1.327997})$. 
\end{theorem}

Our result shows yet another application of the Min-Plus product to an independently-studied problem, ultimately showing a dependence of the complexity of dynamic range mode on the complexity of fast matrix multiplication. Further, in contrast to many other reductions to Min-Plus in which we must assume a structured input on the original problem~\cite{Seidel95,Shoshan99,Zwick02,Chan10,Yuster09,Williams20}, our algorithm works on the fully general dynamic range mode problem. In this sense, our result is perhaps most directly comparable to the batch range mode reduction of Williams and Xu~\cite{Williams20} and the Language Edit Distance, RNA-folding, and Optimum Stack Generation reductions of Bringmann et al.~\cite{Bringmann17}.

\subsection{Discussion of Technical Difficulty}

Despite the new $\tilde{O}(n^{1.4854})$ time algorithm for batch range mode \cite{Williams20}, we cannot directly apply the result to dynamic range mode. The main issue is the element deletion operation. In the range mode algorithm of Williams and Xu (and in many other range mode algorithms), critical points are chosen evenly distributed in the array, and the algorithm precomputes the range mode of intervals between every pair of critical points. In \cite{Williams20}, the improvement is achieved via a faster precomputation algorithm, which uses a Min-Plus product algorithm for structured matrices. However, if element deletion is allowed, the results stored in the precomputation will not be applicable. For example, an interval between two critical points could contain $x$ copies of element $a$, $x-1$ copies of element $b$, and many other elements with frequencies less than $x-1$. During precomputation, the range mode of this interval would be $a$. However, if we delete two copies of $a$, there is no easy way to determine that the mode of this interval has now changed to $b$.

We overcome this difficulty by introducing the Min-Plus-Query problem, as defined in Section~\ref{ourresults}. Intuitively, in the Min-Plus-Query problem, a large portion of the work of the Min-Plus product is put off until the query. It also supports more flexible queries. Using the Min-Plus-Query problem as a subroutine, we will be able to query the most frequent element excluding a set $S$ of forbidden elements. For instance, in the preceding example, we would be able to query the most frequent element that is not $a$. This is the main technical contribution of the paper. 

Another major difference between our algorithm for dynamic range mode and the batch range mode algorithm of Williams and Xu~\cite{Williams20} is the need for rectangular matrix multiplication. In our algorithm, we treat 
elements that appear more than about $N^{2/3}$ times differently from the rest (a similar treatment is given in the dynamic range mode algorithm of Hicham et al.~\cite{zhms2018}).
However, the number of critical points we use is about $N^{1/3}$; thus the number of critical points and frequent elements differ. This contrasts with batch range mode, where elements that appear more than about $\sqrt{n}$ times are considered frequent and the number of critical points used coincides with the number of frequent elements. The consequence of this difference is that a rectangular matrix product is required for dynamic range mode, while a square matrix product sufficed in~\cite{Williams20}.

%% file: related.tex
The range mode problem was first studied formally by Krizanc et al.~\cite{Krizanc05}. They study space-efficient data structures for static range mode, achieving a time-space tradeoff of $O(n^{2-2\epsilon})$ space and $O(n^\epsilon \log n)$ query time for any $0 < \epsilon \leq 1/2$. They also give a solution occupying $O(n^2 \log \log n/\log n)$ space with $O(1)$ time per query.

Chan et al.~\cite{C14} also study static range mode, focusing on linear space solutions. They achieve a linear space data structure supporting queries in $O(\sqrt{n})$ time via clever use of arrays, which can be improved to $O(\sqrt{n/\log n})$ time via bit-packing tricks. Their paper also introduces the conditional lower bound which reduces multiplication of two $\sqrt{n} \times \sqrt{n}$ boolean matrices to $n$ range mode queries on an array of size $O(n)$. As mentioned, combined with the presumed hardness of the online matrix vector problem~\cite{Henzinger15}, this result indicates a dynamic range mode data structure must take greater than $O(n^{1/2-\epsilon})$ for $\epsilon > 0$ time per operation.
Finally, Chan et al.~\cite{C14} also give the first data structure for dynamic range mode. At linear space, their solution achieves $O(n^{3/4}\log n/ \log \log n)$ worst-case time per query and $O(n^{3/4} \log \log n)$ amortized expected time update, and at $O(n^{4/3})$ space, their solution achieves $O(n^{2/3} \log n/ \log \log n)$ worst-case time query and amortized expected update time.

Recently, Hicham et al.~\cite{zhms2018} improved the runtime of dynamic range mode to worst-case $O(n^{2/3})$ time per operation while simultaneously improving the space usage to linear. Prior to this paper, this result was the fastest data structure for dynamic range mode.

A cell-probe lower bound for static range mode has been devised by Greve et al.~\cite{GJLT10}. Their result states that any range mode data structure that uses $S$ memory cells of $w$-bit words needs $\Omega(\frac{\log n}{\log(Sw/n)})$ time to answer a query.

Via reduction to a structured Min-Plus product, Williams and Xu~\cite{Williams20} recently showed that $n$ range mode queries on a fixed array of size $n$ can be answered in $\tilde{O}(n^{1.4854})$ time. Williams and Xu actually show how to compute the \emph{frequency} of the mode for each query. We can adapt this method to find the element that is mode using the following binary search. For query $[l, r]$, we ask the frequency of the mode in range $[l, (l+r)/2]$. If it is the same, we repeat the search with right endpoint in range $[(l+r)/2, r]$; if it is not, we repeat the search with right endpoint in range $[l, (l+r)/2]$. Using this method, we can binary search until we determine when the frequency of the mode changes, thus finding the element that is mode in an additional $O(\log n)$ queries. The algorithm of Williams and Xu can also be used to speed up the preprocessing time of the $O(n)$ space, $O(\sqrt{n})$ query time static range mode data structure to $\tilde{O}(n^{1.4854})$ time.

Both static and dynamic range mode have been studied in approximate settings~\cite{BKMT05,GJLT10,Zein19}.

%% file: prelim.tex
We formally define the Min-Plus product problem and the dynamic range mode problem. 

\begin{problem}[Min-Plus product]
The Min-Plus product of an $m \times n$ matrix $A$ and an $n\times p$ matrix $B$ is the $m\times p$ matrix $C=A\star B$ such that $C_{i,j}=\min_k \{A[i,k]+B[k,j]\}$.
\end{problem}

\begin{problem}[Dynamic Range Mode]
In the dynamic range mode problem, we are given an initially empty sequence and must support the following operations: 
\begin{itemize}
    \item Insert an element at a given position of the sequence. 
    \item Delete one element of the sequence. 
    \item Query the most frequent element of any contiguous subsequence. If there are multiple answers, output any. 
\end{itemize}
It is guaranteed that the size of the array does not exceed $N$ at any point in time. 
\end{problem}

We use $\omega$ to denote the square matrix multiplication exponent, i.e. the smallest real number such that two $n\times n$ matrices can be multiplied in $n^{\omega+o(1)}$ time. The current bound on $\omega$ is $2\leq \omega<2.373$ \cite{legallmult,vmmult}. In this work, we will use fast rectangular matrix multiplication. Analogous to the square case, we use $\omega(k)$ to denote the exponent of rectangular matrix multiplication, i.e., the smallest real number such that an $n\times n^k$ matrix and an $n^k\times n$ matrix can be multiplied in $n^{\omega(k)+o(1)}$ time. Le Gall and Urrutia \cite{gall2018improved} computed smallest upper bounds to date for various values of $k$. In this work, we are mostly interested in values of $\omega(k)$ listed in Figure~\ref{fig:rect_MM}.
\begin{figure}[h]
    \centering
    \begin{tabular}{| c|c|}
    \hline
       $k$  & Upper Bound on $\omega(k)$ \\ 
    \hline
       $1.75$ &  $3.021591$  \\
    \hline
       $2$ &  $3.251640$  \\
    \hline
    \end{tabular}
    \caption{Upper bounds for the exponent of multiplying an $n\times n^k$ matrix and an $n^k\times n$ matrix \cite{gall2018improved}.}
    \label{fig:rect_MM}
\end{figure}

It is known that the function $\omega(k)$ is convex for $k>0$ (see e.g. \cite{le2012faster}, \cite{lotti1983asymptotic}), so we can use values of $\omega(p)$ and $\omega(q)$ to give upper bounds for $\omega(k)$ as long as $p \le k \le q$. 
\begin{fact}
\label{fact:rect_MM_convex}
When $0 < p \le k \le q$, $\omega(k) \le \frac{k-p}{q-p} \omega(q) + \frac{q-k}{q-p} \omega(p)$.
\end{fact}
Combining Figure~\ref{fig:rect_MM} and Fact~\ref{fact:rect_MM_convex}, we obtain the following bound on $\omega(k)$ when $k \in [1.75, 2]$. 
\begin{corollary}
\label{cor:rect_mult}
When $1.75 \le k \le 2$, $\omega(k) \le 0.920196 k + 1.41125$. 
\end{corollary}

%% file: algo.tex
A main technical component for our dynamic range mode algorithm is the use of the \textit{Min-Plus-Query} problem, which is formally defined in Section~\ref{sec:intro}. We are given two matrices $A, B$. For each query, we are given three parameters $i, j, S$,  and we need to compute $\min_{k \not \in S} \{A_{i,k}+B_{k,j}\}$.

If we just use the Min-Plus-Query problem, we can only compute the frequency of the range mode. Although we can binary search for the most frequent element as described in Section~\ref{sec:related}, we are also able to return the witness from the Min-Plus-Query problem organically. This construction may be of independent interest.
\begin{problem}[\textit{Min-Plus-Query-Witness} problem]
    During initialization, we are given two matrices $A, B$. For each query, we are given three parameters $i, j, S$, where $i, j$ are two integers, and $S$ is a set of integers. We must output an index $k^* 
    \notin S$ such that  $A_{i,k^*}+B_{k^*,j} = \min_{k \not \in S} \{A_{i,k}+B_{k,j}\}$.
\end{problem}

If $A$ is an $n \times n^{s}$ matrix and $B$ is an $n^{s} \times n$ matrix, then the naive algorithm for Min-Plus-Query just enumerates all possible indices $k$ for each query, which takes $O(n^{s})$ time per query. In order to get a faster algorithm for dynamic range mode, we need to achieve $\tilde{O}(n^{2+s - \epsilon})$ preprocessing time and $\tilde{O}(n^{s - \epsilon} + |S|)$ query time, for some $\epsilon>0$, where $A, B$ are some special matrices generated by the range mode instance. Specifically, matrix $B$ meets the following two properties:
\begin{enumerate}
    \item Each row of $B$ is non-increasing;
    \item The difference between the sum of elements in the $j$-th column and the sum of elements in the $(j+1)$-th column is at most $n^s$, for any $j$.
\end{enumerate}
Williams and Xu \cite{Williams20} give a faster algorithm for multiplying an arbitrary matrix $A$ with such matrix $B$, which leads to a faster algorithm for static range mode. We will show that  nontrivial data structures exist for the Min-Plus-Query problem for such input matrices $A$ and $B$. Such a data structure will lead to a faster algorithm for dynamic range mode.

In the following lemma, we show a data structure for the Min-Plus-Query problem when both input matrices have integer weights small in absolute value. 

\begin{lemma}
\label{lem:both_small}
Let $s \ge 1$ be a constant. Let $A$ and $B$ be two integer matrices of dimension $n \times n^s$ and $n^s \times n$, respectively, with entries in $\{-W, \ldots, W\} \cup \{\infty\} $ for some $W \ge 1$. Then we can solve the Min-Plus-Query problem of $A$ and $B$ in  $\tilde{O}(Wn^{\omega(s)})$ preprocessing time and $\tilde{O}(|S|)$ query time. The space complexity is $\tilde{O}(Wn^2 + n^{1+s})$. 
\end{lemma}
\begin{proof}
The algorithm uses the idea by Alon, Galil and Margalit in \cite{AlonGM97}, which computes the Min-Plus product of $A, B$ in $\tilde{O}(Wn^{\omega(s)}) $ time. 

In their algorithm, they first construct matrix $A'$ defined by 
\[
A'_{i,k}=\left\{
  \begin{array}{ll}
    (n^s+1)^{A_{i,k}+W} & \text{if } A_{i,k} \ne \infty,\\
    0 & \text{otherwise}.
  \end{array}
\right.
\]
We can define $B'$ similarly. Then the product $A'B'$ captures some useful information about the Min-Plus product of $A$ and $B$. Namely, for each entry $(A' B')_{i,j}$, we can uniquely write it as $\sum_{t \ge 0} r^{i,j}_t (n^s+1)^t$ for integers $0 \le r^{i,j}_t \le n^s$. Note that $r^{i,j}_t$ exactly equals the number of $k$ such that $A_{i,k}+B_{k,j} = t-2W$. Thus, we can use $A'B'$ to compute the Min-Plus Product of $A$ and $B$. 

In our algorithm, we use a range tree to maintain the sequence $r^{i,j}_t$ for each pair of $i, j$. The preprocessing takes $\tilde{O}(Wn^{\omega(s)})$ time, which is the time to compute $A'B'$ and the sequences $r^{i,j}_t$. 

During each query, we are given $i, j, S$. We enumerate each $k \in S$, and decrement $r^{i,j}_{A_{i,k}+B_{k,j}+2W}$ in the range tree if $A_{i,j}+B_{k,j} < \infty$. After we do this for every $k \in S$, we query the range tree for the smallest $t$ such that $r^{i,j}_t \ne 0$, so $t-2W$ is the answer to the Min-Plus-Query query. After each query, we need to restore the values of $r^{i,j}$, which can also be done efficiently. The query time is $\tilde{O}(|S|)$, since each update and each query of range tree takes $\tilde{O}(1)$ time. The space complexity should be clear from the algorithm. 
\end{proof}

In the previous lemma, the data structure only answers the Min-Plus-Query problem. In all subsequent lemmas, the data structure will be able to handle the Min-Plus-Query-Witness problem. 

In the next lemma, we use Lemma~\ref{lem:both_small} as a subroutine to show a data structure for the Min-Plus-Query-Witness problem when only matrix $A$ has small integer weights in absolute value. 


\begin{lemma}
\label{lem:A_small}
Let $s \ge 1$ be a constant. Let $A$ and $B$ be two integer matrices of dimension $n \times n^s$ and $n^s \times n$, respectively, where $A$ has entries in $\{-W, \ldots, W\} \cup \{\infty\} $ for some $W \ge 1$, and $B$ has arbitrary integer entries represented by $\polylog n$ bit numbers. 
Then for every integer $1 \le P \le n^s$, 
 we can solve the Min-Plus-Query-Witness problem of $A$ and $B$ in $O(\frac{n^s}{P} W n^{\omega(s)})$ preprocessing time and $O(|S| + P)$ query time. The space complexity is $\tilde{O}(\frac{Wn^{2+s}}{P} + \frac{n^{1+2s}}{P})$. 
\end{lemma}
\begin{proof}

For simplicity, assume $P$ is a factor of $n^s$. We sort each column of matrix $B$ and put entries whose rank is between $(\ell-1)P + 1$ and $\ell P$ into the $\ell$-th bucket. We use $K_{j, \ell}$ to denote the set of row indices of entries in the $\ell$-th bucket of the column $j$. We use $L_{j,\ell}$ to denote the smallest entry value of the bucket $K_{j, \ell}$, and use $H_{j, \ell}$ to denote the largest entry value. Formally, 
\[L_{j, \ell} = \min_{k \in K_{j, \ell}} B_{k, j} \quad \text{and} \quad H_{j, \ell} = \max_{k \in K_{j, \ell}} B_{k, j}.\]

For each $\ell \in [n^s / P]$, we do the following\footnote{We use $[n]$, with $n$ integer, to denote the set $\{1, 2, \ldots, n\}$.}. We create an $n^s \times n$ matrix $B^\ell$ and initialize all its entries to $\infty$. Then for each column $j$, if $H_{j,\ell} - L_{j, \ell} \le 2W$ (we will call it a small bucket), we set $B^\ell_{k, j} := B_{k,j} - L_{j, \ell}  - W$ for all $k \in K_{j, \ell}$. We will handle the case $H_{j,\ell} - L_{j, \ell} > 2W$ (large bucket) later. Clearly, all entries in $B^\ell$ have values in $\{-W, \ldots, W\} \cup \{\infty\}$, so we can use the algorithm in Lemma~\ref{lem:both_small} to preprocess $A$ and $B^\ell$ and store the data structure in $D^\ell$. Also, for each pair $(i, j)$, we create a range tree $\mathcal{T}_{\text{small}}^{i,j}$ on the sequence $(A  \star B^1)_{i,j}, (A  \star B^2)_{i,j}, (A  \star B^3)_{i,j}, \ldots$, $(A  \star B^{n^s/P})_{i,j}$, which stores the optimal Min-Plus values when $k$ is  from a specific small bucket. 
This part takes $\tilde{O}(\frac{n^s}{P} W n^{\omega(s)})$ time. The space complexity is $\frac{n^s}{P}$ times more than the space complexity of Lemma~\ref{lem:both_small}, so space complexity of this part is $\tilde{O}(\frac{Wn^{2+s}}{P} + \frac{n^{1+2s}}{P})$.

We also do the following preprocessing for buckets where $H_{j,\ell} - L_{j, \ell} > 2W$. We first create a $0/1$ matrix $\bar{A}$ where $\bar{A}_{i,k} = 1$ if and only if $A_{i, k} \ne \infty$. Then for each $\ell \in [n^s/ P]$, we create a $0/1$ matrix $\bar{B}^\ell$ such that $\bar{B}^\ell_{k,j} = 1$ if and only if $k \in K_{j,\ell}$ and $H_{j,\ell} - L_{j, \ell} > 2W$. Then we use fast matrix multiplication to compute the product $\bar{A}\bar{B}^\ell$. If $K_{j, \ell}$ is a large bucket, the $(i,j)$-th entry of $\bar{A}\bar{B}^\ell$ is the number of $k \in K_{j,\ell}$ such that $A_{i,k} < \infty$; if $K_{j, \ell}$ is a small bucket, the $(i,j)$-th entry is $0$. For each pair $(i, j)$, we create a range tree $\mathcal{T}_{\text{large}}^{i,j}$ on the sequence $(\bar{A}\bar{B}^1)_{i,j}, (\bar{A}\bar{B}^2)_{i,j}, (\bar{A}\bar{B}^3)_{i,j}, \ldots, (\bar{A}\bar{B}^{n^s/P})_{i,j}$. This part takes $\tilde{O}(\frac{n^s}{P} n^{\omega(s)})$ time, which is dominated by the time for small buckets. The space complexity is also dominated by the data structures for small buckets.

Now we describe how to handle a query $(i, j, S)$. First consider small buckets. In $O(|S|)$ time, we can compute the set of small buckets $K_{j, \ell}$ that intersect with $S$. For each such $K_{j, \ell}$, we can query the data structure $D^\ell$ with input $(i, j, S \cap K_{j, \ell})$ to get the optimum value when $k \in K_{j, \ell}$. For each small bucket that intersects with $S$, we can set its corresponding value in the range tree $\mathcal{T}_{\text{small}}^{i, j}$ to $\infty$, then we can compute the optimum value of all small buckets that do not intersect with $S$ by querying the minimum value of the range tree $\mathcal{T}_{\text{small}}^{i, j}$. After this query, we need to restore all values in the range tree. 
It takes $\tilde{O}(|S|)$ time to handle small buckets on query.

Now consider large buckets. 
Intuitively, we want to enumerate indices in all large buckets $K_{j,\ell}$ such that there exists  an index $k \in K_{j, \ell} \cap ([n^s] \setminus S)$ where $A_{i,k} < \infty$. However, doing so would be prohibitively expensive. We will show that we only need two such buckets. 
Consider three large buckets $l_1 < l_2 < l_3$. 
Pick any $k_1 \in T_{j, l_1}, k_3 \in T_{j, l_3}$ such that $A_{i, k_1} < \infty$. Since
\[A_{i, k_1} + B_{k_1, j} \le W+ L_{j, l_2} < W+H_{j, l_2} - 2W < A_{i, k_3} + B_{k_3, j}, \]
$k_3$ can never be the optimum. Thus, it suffices to find the smallest two buckets such that there exists  an index $k \in K_{j, \ell} \cap ([n^s] \setminus S)$ where $A_{i,k} < \infty$, and then enumerate all indices in these two buckets. To find such two buckets, we can enumerate over all indices $k \in S$, and if $A_{i,k} < \infty$ we can decrement the corresponding value in the range tree  $\mathcal{T}_{\text{large}}^{i,j}$. Thus, we can  compute the two smallest buckets by querying the two earliest nonzero values in the range tree. We also need to restore the range tree after the query. The range tree part takes $\tilde{O}(|S|)$ time and scanning the two large buckets requires $O(P)$ time. 
Thus, 
this step takes $\tilde{O}(|S| + P)$ time. 

At this point, we will know the bucket that contains the optimum index $k^*$. Thus, we can iterate all indices in this bucket to actually get the witness for the Min-Plus-Query-Witness query. It takes $O(P)$ time to do so. 

In summary, the preprocessing time, query time, and space complexity meet the promise in the lemma statement. 

\end{proof}

In the following lemma, we show a data structure for the Min-Plus-Query-Witness problem when the matrix $B$ has the \textit{bounded difference} property, which means that nearby entries in each row have close values. The proof adapts the strategy of \cite{Williams20}. 

\begin{lemma}
\label{lem:bounded_diff_B}
Let $s \ge 1$ be a constant. Let $A$ be an $n \times n^s$ integer matrix, and let $B$ be an $n^s \times n$ integer matrix. It is guaranteed that there exists $1 \le \Delta \le \min\{n, W\}$, such that for every $k$, $|B_{k, j_1} - B_{k, j_2}| \le W$ as long as $\lceil j_1 / \Delta \rceil = \lceil j_2/ \Delta \rceil$.  
Then for every $L = \Omega(\Delta)$,
 we can solve the Min-Plus-Query-Witness problem of $A$ and $B$ in $\tilde{O}(\Delta^2 \frac{n^s}{L} W n^{\omega(s)} + \frac{n^{2+s}}{\Delta})$ preprocessing time and $\tilde{O}(L)$ query time, when $|S| < L$. The space complexity is $\tilde{O}(\frac{\Delta^2 W n^{2+s}}{L} + \frac{\Delta^2 n^{1+2s}}{L} + \frac{n^{2+s}}{\Delta})$.

\end{lemma}

\begin{proof}

\textbf{Preprocessing Step 1: Create an Estimation Matrix}\vspace{1pt}

First, we create a matrix $\hat{B}$, where $\hat{B}_{k, j} = B_{k, \lceil  j / \Delta \rceil \Delta}$. By the property of matrix $B$, $|\hat{B}_{k, j} - B_{k, j}| \le W$ for every $k, j$. For each pair $(i, j)$, we compute the $L$-th smallest value of $A_{i,k}+\hat{B}_{k,j}$ among all $1 \le k \le n^s$, and denote this value by $\hat{C}^L_{i,j}$. Notice that $\hat{C}^L_{i,j} = \hat{C}^L_{i, \lceil  j / \Delta \rceil \Delta}$, so it suffices to compute $\hat{C}^L_{i,j}$ when $j$ is a multiple of $\Delta$, and we can infer other values correctly. It takes $O(n^s)$ time to compute each $\hat{C}^L_{i,j}$, so this step takes $O(n^{2+s}/\Delta)$ time. 

If we similarly define $C^L_{i,j}$ as the $L$-th smallest value of $A_{i,k}+B_{k,j}$ among all $1 \le k \le n^s$, then $|C^L_{i,j}-\hat{C}^L_{i,j}| \le W$ by the following claim, whose proof is omitted for space constraint. 

\begin{claim*}
Given two sequences $(a_k)_{k=1}^m$ and $(b_k)_{k=1}^m$ such that $|a_k - b_k| \le W$, then the $L$-th smallest element of $a$ and the the $L$-th smallest element of $b$ differ by at most $W$. 
\end{claim*}

Also, in $\tilde{O}(n^{2+s}/\Delta)$ time, we can compute a sorted list $\mathcal{L}^{i, j}_{\text{small}}$ of indices $k$ sorted by the value $A_{i,k}+\hat{B}_{k,j} - \hat{C}^L_{k,j}$, for every $i$, and every $j$ that is a multiple of $\Delta$.

The space complexity in this step is not dominating.\medskip

\noindent\textbf{Preprocessing Step 2: Perform Calls to Lemma~\ref{lem:A_small}}\vspace{1pt}

For some integer $\rho \ge 1$, we will perform $\rho$ rounds of the following algorithm. At the $r$-th round for some $1 \le r \le \rho$, we randomly sample $j^r \in [n]$, and let $A^r_{i,k} := A_{i,k} + B_{k, j^r} - \hat{C}^L_{i,j^r}$ and $B^{r}_{k,j}:= B_{k,j} - B_{k, j^r}$. 
Clearly, $A^r_{i,k} + B^r_{k,j} = A_{i,k}+B_{k,j} - \hat{C}^L_{i,j^r}$. 
For each pair $(i, k)$, we find the smallest $r$ such that $|A^r_{i,k}| \le 3W$. We keep these entries as they are and replace all other entries by $\infty$. For every $(i, k)$, there exists at most one $r$ such that $A^r_{i,k} \ne \infty$.
Then we use Lemma~\ref{lem:A_small} to preprocess $A^r$ and $B^r$ for every $1 \le r \le \rho$. Thus, this part  takes $O(\rho \frac{n^s}{P} W n^{\omega(s)})$ time, for some integer $P$ to be determined later. Note that this parameter also affects the query time. This step stores $\rho$ copies of the data structure from Lemma~\ref{lem:A_small}, so the space complexity is $\tilde{O}(\rho \frac{W n^{2+s}}{P} + \rho \frac{n^{1+2s}}{P})$. 

Note that this step is the only step that uses randomization. We can use the method of \cite{Williams20}, Appendix A, to derandomize it. We omit the details for simplicity.\medskip

\noindent\textbf{Preprocessing Step 3: Handling Uncovered Pairs}\vspace{1pt}

For a pair $(i, k)$, 
if $A^r_{i, k} \ne \infty$  for any $r$, we call $(i, k)$ \textit{covered}; otherwise, we call the pair $(i, k)$ \textit{uncovered}. For each pair $(i, j)$, we enumerate all $k$ such that $|A_{i,k}+\hat{B}_{k, j} - \hat{C}^L_{i,j}| \le 2W$ and $(i, k)$ is uncovered. Notice that since $A_{i,k}+\hat{B}_{k, j} - \hat{C}^L_{i,j} = A_{i,k}+\hat{B}_{k, \lceil j / \Delta \rceil \Delta} - \hat{C}^L_{i,\lceil j / \Delta \rceil \Delta}$, we only need to exhaustively enumerate all $k \in [n^s]$ when $j$ is a multiple of $\Delta$. Thus, if the total number of  $(i, k, j)$ where $|A_{i,k}+\hat{B}_{k, j} - \hat{C}^L_{i,j}| \le 2W$ and $(i, k)$ is uncovered is $X$, then we can enumerate all such triples $(i, k, j)$ in $O(X+n^{2+s}/\Delta)$ time. 

It remains to bound the total number of triples that satisfy the condition. Fix an arbitrary pair $(i, k)$, and suppose the number of $j$ such that $|A_{i,k}+\hat{B}_{k, j} - \hat{C}^L_{i,j}| \le 2W$ is at least $(10+s)n \ln n/\rho$. Then with probability at least $1-(1-\frac{(10+s) \ln n}{\rho})^\rho \ge 1-\frac{1}{n^{10+s}}$, we pick a $j^r$ where $|A_{i,k}+\hat{B}_{k, j^r} - \hat{C}^L_{i,j^r}| \le 2W$. Therefore,
\[|A^r_{i,k}| = \left| A_{i,k} + B_{k, j^r} - \hat{C}^L_{i,j^r} \right| \le \left| A_{i,k} + \hat{B}_{k, j^r} - \hat{C}^L_{i,j^r} \right| + \left|\hat{B}_{k, j^r} - B_{k, j^r} \right| \le 3W,\]
which means $(i, k)$ is covered. Therefore, with high probability, all pairs of $(i, k)$ where the number of $j$ such that $|A_{i,k}+\hat{B}_{k, j} - \hat{C}^L_{i,j}| \le 2W$ is at least $(10+s)n \ln n/\rho$ will be covered. In other words, $X = O(n^{1+s} \cdot n \ln n / \rho)=\tilde{O}(n^{2+s} /\rho)$. 

For each pair $(i, j)$, if we enumerate more than $L$ indices $k$, we only keep the $L$ values of $k$ that give the smallest values of $A_{i,k}+B_{k,j}$. We call this list $\mathcal{L}^{i,j}_{\text{triple}}$. From previous discussion, the time cost in this step is $\tilde{O}(n^{2+s}/\rho +n^{2+s}/\Delta)$. Since we need to store all the triples, the space complexity is $O(n^{2+s}/\rho)$.\medskip

\noindent\textbf{Handling Queries}\vspace{1pt}

Now we discuss how to handle queries. 
For each query $(S, i, j)$, let $k^* =\argmin_{k \not \in S} A_{i,k}+B_{k,j}$ be the optimum index. Consider two cases:
\begin{itemize}
\item $(i, k^*)$ is covered. By definition of being covered, there exists a round $r$ such that $A^r_{i,k^*} = A_{i,k^*}+B_{k^*,j^r}- \hat{C}^L_{i,j^r}$, so $A^r_{i, k^*}+B^r_{k^*, j}=A_{i, k^*}+B_{k^*,j}-\hat{C}^L_{i, j^r}$. Therefore, we can query the data structure in Lemma~\ref{lem:A_small} for every $A^r$ and $B^r$ and denote $b^r$ as the result. The answer is given by the smallest value of $b^r+\hat{C}^L_{i,j^r}$ over all $r$. The witness is given by the data structure of Lemma~\ref{lem:A_small}. 

Note that when querying $A^r$ and $B^r$, we only need to pass the set $\{k \in S: A^r_{i, k} \ne \infty\}$. For every $k \in S$, there is at most one $r$ such that $A^r_{i,k} \ne \infty$, 
so the total size of the sets passing to the data structure of Lemma~\ref{lem:A_small} is $|S|$. Thus, this case takes $O(|S| + \rho P)$ time. 

\item  $(i, k^*)$ is uncovered. There are still two possibilities to consider in this case. 

\begin{itemize}
\item \textbf{Possibility I}: $A_{i,k^*}+\hat{B}_{k^*, j} - \hat{C}^L_{i,j} < -2W$. In this case, 
\[A_{i,k^*}+B_{k^*, j} \le A_{i,k^*}+\hat{B}_{k^*, j} + W < \hat{C}^L_{i,j} - W,\]
so the optimum value is smaller than $\hat{C}^L_{i,j}$. By reading the list $\mathcal{L}^{i, \lceil j/\Delta\rceil \Delta}_{\text{small}}$, we can effectively find all such $k$ where $A_{i,k}+\hat{B}_{k, j} - \hat{C}^L_{i,j} < -2W$ in time linear to the number of such $k$. The number of such $k$ is at most $L$, by the definition of $\hat{C}^L_{i,j}$. Thus, this part takes $O(L)$ time. 

\item \textbf{Possibility II}: $A_{i,k^*}+\hat{B}_{k^*, j} - \hat{C}^L_{i,j} \ge -2W$. In fact, in this case, we further have
\[A_{i,k^*}+\hat{B}_{k^*,j} - \hat{C}^{L}_{i, j}\le  A_{i,k^*}+B_{k^*,j} - C^L_{i,j} +2W \le 2W,\]
where $A_{i,k^*}+B_{k^*,j} - C^L_{i,j} \le 0$  because $|S| < L$. Therefore, in this case, we have $|A_{i,k^*}+\hat{B}_{k^*, j} - \hat{C}^L_{i,j}| \le 2W$, so
we can enumerate  all indices in $\mathcal{L}^{i, j}_{\text{triple}}$ and take the best choice. 
This takes $O(L)$ time.
\end{itemize}
\end{itemize}

\noindent\textbf{Time and Space Complexity}\vspace{1pt}

In summary, the preprocessing time is \[\tilde{O}\left(\rho \frac{n^s}{P} W n^{\omega(s)} + n^{2+s}/\Delta + n^{2+s}/\rho\right),\]
and the query time is 
$\tilde{O}(L + \rho P)$.
To balance the terms, we can set $\rho = \Delta$ and $P = \frac{L}{\Delta}$ to achieve a $\tilde{O}(\Delta^2 \frac{n^s}{L} W n^{\omega(s)} + \frac{n^{2+s}}{\Delta})$ preprocess time and a $\tilde{O}(L)$ query time. Note that since we need $P \ge 1$, we must have $L = \Omega(\Delta)$. 

From the preprocessing steps, the space complexity is $\tilde{O}(\rho \frac{W n^{2+s}}{P} + \rho \frac{n^{1+2s}}{P} + n^{2+s} / \rho)$. Plugging in $\rho = \Delta$ and $P = \frac{L}{\Delta}$ reduces this to
\[\tilde{O}\left(\frac{\Delta^2 W n^{2+s}}{L} + \frac{\Delta^2 n^{1+2s}}{L} + \frac{n^{2+s}}{\Delta}\right),\]
as given in the statement of the lemma.
\end{proof}

The next lemma is our last data structure for Min-Plus-Query-Witness problems. 

\begin{lemma}\label{lem:lem_multi_inc}
Let $s \ge 1$ be a constant. Let $A$ be an $n \times n^s$ integer matrix and $B$ be  an $n^s \times n$ integer matrix. Suppose matrix $B$ satisfies
\begin{enumerate}
    \item Each row of $B$ is non-increasing;
    \item The difference between the sum of elements in the $j$-th column and the sum of elements in the $(j+1)$-th column is at most $n^s$, for any $j$.
\end{enumerate}
Then for every positive integer $L = \Omega(n^{\omega(s) - 2})$, 
 we can solve the Min-Plus-Query-Witness problem of $A$ and $B$ in $\tilde{O}(n^{\frac{8}{5}+s+\frac{1}{5} \omega(s)} L^{-\frac{1}{5}})$ preprocessing time and $\tilde{O}(L)$ query time, when $|S| < L$. The space complexity is $\tilde{O}(L^{-\frac{1}{5}} n^{\frac{18}{5} + s - \frac{4}{5}\omega(s) }+L^{-\frac{3}{5}} n^{\frac{9}{5} + 2s - \frac{2}{5}\omega(s)}  + L^{-\frac{1}{5}} n^{\frac{8}{5} + s + \frac{1}{5}\omega(s)}).$
 
\end{lemma}

\begin{proof}

Let $\Delta, W \ge 1$ be small polynomials in $n$ to be fixed later. Define $I(j)$ to be the interval $[j-\Delta+1, j]$. 

Let $j'$ be any multiple of $\Delta$. By property \textbf{2} of matrix $B$, $\sum_{k=1}^{n^s} B_{k,j} - \sum_{k=1}^{n^s} B_{k,j+1} \le n^s$ for any $j \in I(j')$. Thus, we have
\[\sum_{k=1}^{n^s} B_{k,j
'-\Delta+1} - \sum_{k=1}^{n^s} B_{k,j'}\le \Delta n^s. \]
By averaging, there are at most $\Delta n^s / W$ indices $k \in [n^s]$ such that $B_{k,j
'-\Delta+1}-B_{k,j'} > W$. We create a new matrix $\hat{B}$, initially the same as matrix $B$. For each $k$ such that $B_{k,j'
-\Delta+1}-B_{k,j'}  > W$, and for each $j \in I(j')$, we set $\hat{B}_{k,j}$ as $M$, where $M$ is some large enough integer. After this replacement, $\hat{B}_{k,j'
-\Delta+1}-\hat{B}_{k,j'}  \le W$ for any $k$ and any $j'$ multiple of $\Delta$. Also, since $\hat{B}_{k,j'
-\Delta+1} \ge \hat{B}_{k,j} \ge \hat{B}_{k,j'}$ for any $j \in I(j')$, we have that $|\hat{B}_{k,j_1} - \hat{B}_{k,j_2}| \le W$ as long as $\lceil j_1 / \Delta \rceil = \lceil j_2 / \Delta \rceil$. Therefore, we can use Lemma~\ref{lem:bounded_diff_B} to preprocess $A$ and $\hat{B}$ in $O(\Delta^2 \frac{n^s}{L} W n^{\omega(s)} + \frac{n^{2+s}}{\Delta})$ time. The space complexity is $\tilde{O}(\frac{\Delta^2 W n^{2+s}}{L} + \frac{\Delta^2 n^{1+2s}}{L} + \frac{n^{2+s}}{\Delta})$.

On the other hand, note that $\hat{B}$ differs with $B$ on at most $n^{1+s} \Delta / W$ entries, so we need to do some extra preprocessing to handle those entries. For each pair $(i, j)$, we initialize a range tree $\mathcal{T}^{(i,j)}$ whose elements are all $\infty$ (it takes $\tilde{O}(1)$ time to initialize each range tree if we implement it carefully). Then for every $k$ such that $B_{k,j} \ne \hat{B}_{k, j}$, we set the $k$-th element in $\mathcal{T}^{(i,j)}$ as $A_{i,k}+B_{k,j}$. 
The total number of operations we perform in all the range trees are $O(n^{2+s} \Delta / W)$, so this part takes $\tilde{O}(n^{2+s} \Delta / W)$ time. The space complexity is also $\tilde{O}(n^{2+s} \Delta / W)$.

During a query $(S, i, j)$, we first query the data structure in Lemma~\ref{lem:bounded_diff_B} on matrix $A$ and $\hat{B}$ with parameters $(S, i, j)$. Then we query the minimum value from the range tree $\mathcal{T}^{(i,j)}$ after setting all $A_{i,k}+B_{k,j}$ as $\infty$ for $k \in S$. Taking the minimum of these two queries gives the answer. The optimum index $k^*$ is either given by the data structure of Lemma~\ref{lem:bounded_diff_B} or can be obtained from the range tree.

Thus, the preprocessing time of the algorithm is 
\[\tilde{O}(\Delta^2 \frac{n^s}{L} W n^{\omega(s)} + \frac{n^{2+s}}{\Delta}+n^{2+s} \Delta / W), \]
and the query time is $\tilde{O}(L)$.
We get the desired preprocessing time by setting $\Delta = L^{1/5}n^{\frac{2}{5}-\frac{1}{5} \omega(s)} $ and $W = \Delta^2$. Since we need $\Delta \ge 1$, we require that $L = \Omega(n^{\omega(s) - 2})$. In Lemma~\ref{lem:bounded_diff_B}, we also requires that $L=\Omega(\Delta)$, but this is always true when $L = \Omega(n^{\omega(s) - 2})$.

By previous discussion, the space complexity is $\tilde{O}(\frac{\Delta^2 W n^{2+s}}{L} + \frac{\Delta^2 n^{1+2s}}{L} + \frac{n^{2+s}}{\Delta}+n^{2+s} \Delta / W)$. Plugging in the value for $\Delta$ and $W$ simplifies the complexity to 
\[
\tilde{O}(L^{-1/5} n^{18/5 + s - 4\omega(s)/5 }+L^{-3/5} n^{9/5 + 2s - 2\omega(s)/5}  + L^{-1/5} n^{8/5 + s + \omega(s) / 5}).
\]

\end{proof}

Finally, we can apply the data structure of Lemma~\ref{lem:lem_multi_inc} to prove Theorem~\ref{maintheorem}.

\begin{proof}[Proof of Theorem \ref{maintheorem}]

For clarity, we will use \textit{element} to refer to a specific item $a_i$ of the sequence and use \textit{value} to refer to all elements of a given type. Given a pointer to an element of the sequence $a_i$, we assume the ability to look up its index $i$ in the sequence in $\tilde{O}(1)$ time by storing all elements of the sequence in a balanced binary search tree with worst-case time guarantees (e.g. a red-black tree). Thus we can go from index $i$ to element $a_i$ and back via appropriate rank and select queries on the balanced binary search tree. We may also add or remove an element $a_i$ from the sequence, and thus the binary search tree, in $\tilde{O}(1)$ time.

Let $T_1, T_2, T_3$ be three parameters of the algorithm. Parameter $T_1$ is a threshold that controls the number of ``frequent'' colors,
 $T_2$ controls how frequently the data structure is rebuilt,
 and $T_3$ represents the size of blocks in the algorithm.

We call values that appear more than $N/T_1$ times \textit{frequent} and all other values \textit{infrequent}. Thus, there are at most $T_1$ frequent values at any point in time. Note that a fixed value can change from frequent to infrequent, or from infrequent to frequent, via a deletion or insertion.\medskip

\noindent\textbf{Infrequent Values}\vspace{1pt}

First, we discuss how to handle infrequent values. We maintain $\frac{N}{T_1}$ balanced search trees $\mathcal{BST}_1, \ldots,  \mathcal{BST}_{\frac{N}{T_1}}$. 
For balanced search tree $\mathcal{BST}_k$, we prepare the key/value pairs in the following way. Fix a given value of the sequence. Say all its occurrences are at indices $i_1, i_2, \ldots, i_t$. Then we insert the key/value pairs $(i_x, i_{x + k - 1})$ to $\mathcal{BST}_k$ for every $1 \le x \le t - k + 1$. However, the indices themselves would need updating when sequence $a$ is updated. Instead of inserting the indices themselves, we insert corresponding pointers to the nodes of the binary search tree that holds sequence $a$. That way we can perform all comparisons using binary search tree operations in $\tilde{O}(1)$ time, without needing to update indices when sequence $a$ changes. We also augment each balanced search tree $\mathcal{BST}_i$ so that every subtree stores the smallest value $y$ of any pair $(x, y)$ in the subtree. After an insertion or deletion, we need to update a total of $O((\frac{N}{T_1})^2)$ pairs. Thus, we can maintain these  balanced search trees in $\tilde{O}((\frac{N}{T_1})^2)$ time per operation.

During a query $[l, r]$, we iterate through all the balanced search trees $\mathcal{BST}_1, \ldots,  \mathcal{BST}_{\frac{N}{T_1}}$. If there exists a pair $(i_1, i_2) \in \mathcal{BST}_k$ such that $l \le i_1 \le i_2 \le r$, then the range mode is at least $k$. Thus, if the range mode is an infrequent value, we can find its frequency and corresponding value by querying the balanced search trees. The query time is $\tilde{O}(\frac{N}{T_1})$, which is not the dominating term.\medskip

\noindent\textbf{Newly Modified Values}\vspace{1pt}

We now consider how to handle frequent values. We handle newly modified values and unmodified values differently. We will rebuild our data structure after every $T_2$ operations, and call values that are inserted or deleted after the last rebuild \textit{newly modified values}. 

For every value, we maintain a balanced search tree of occurrences of this value in the sequence. It takes $\tilde{O}(1)$ time per operation to maintain such balanced search trees. 
Thus, given an interval $[l, r]$, it takes $\tilde{O}(1)$ time to query the number 
of occurrences of a particular value in the interval. We use this method to query the number of occurrences of each newly modified value. Since there can be at most $T_2$ such values, this part takes $\tilde{O}(T_2)$ time per operation.\medskip

\noindent\textbf{Data Structure Rebuild}\vspace{1pt}

It remains to handle the frequent, not newly modified values during each rebuild. In this case, we will assume we can split the whole array roughly equally into a left half and right half. We can recursively build the data structure on these two halves so that we may assume a range mode query interval has left endpoint in the left half and right endpoint in the right half. The recursive construction adds only a poly-logarithmic factor to the complexity. 

We split the left half and the right half into consecutive segments of length at most $T_3$, so that there are $O(N/T_3)$ segments. We call the segments $P_1, P_2, \ldots, P_m$ in the left half and $Q_1, Q_2, \ldots, Q_m$ in the right half, where segments with a smaller index are closer to the middle of the sequence.

Let $v_1, v_2, \ldots, v_l$ be the frequent values during the rebuild. We create a matrix $A$ such that $A_{i,k}$ equals the negation of the number of occurrences of $v_k$ in segments $P_1, \ldots, P_i$; similarly, we create a matrix $B$ such that $B_{k, j}$ equals the negation of the number of occurrences of $v_k$ in segments $Q_1, \ldots, Q_j$. Note that the negation of the value $A_{i,k}+B_{k,j}$ is the frequency of value $v_k$ in the interval from $P_i$ to $Q_j$. 
It is not hard to verify that matrix $B$ satisfies the requirement of Lemma~\ref{lem:lem_multi_inc}.  We take the negation here since Lemma~\ref{lem:lem_multi_inc} handles $(\min, +)$-product instead of $(\max, +)$-product. Then we use the preprocessing part of Lemma~\ref{lem:lem_multi_inc} with matrices $A, B$, and $L=T_2$. If we let $T_1 = N^{t_1}, T_2 = N^{t_2}, T_3 = N^{t_3}$, then in the notation of Lemma~\ref{lem:lem_multi_inc}, $n=m=O(N/T_3)=O(N^{1-t_3})$ and $n^s = O(T_1) = O(N^{t_1})$, so $s=\frac{t_1}{1-t_3}$ and $L = N^{t_2}$. Thus, by Lemma~\ref{lem:lem_multi_inc} the rebuild takes
\[\tilde{O}(N^{(1-t_3)(\frac{8}{5}+\frac{t_1}{1-t_3}+\frac{1}{5}\omega( \frac{t_1}{1-t_3})) - \frac{1}{5}t_2})\]
time. Since we perform the rebuild every $T_2$ operations, the amortized cost of rebuild is
\[\tilde{O}(N^{(1-t_3)(\frac{8}{5}+\frac{t_1}{1-t_3}+\frac{1}{5}\omega( \frac{t_1}{1-t_3})) - \frac{6}{5}t_2})\]
per operation.

Now we discuss how to handle queries for frequent, unmodified elements. 
For a query interval $[l, r]$, we find all the segments inside the interval $[l, r]$. The set of such segments must have the form $P_1, \cup \cdots \cup P_i \cup Q_1 \cup \cdots \cup Q_j$ for some $i, j$. We scan through all elements in $[l, r] \setminus (P_1 \cup \cdots \cup P_i \cup Q_1 \cup \cdots \cup Q_j)$, and use their frequency to update the answer. Since the size of segments is $O(T_3)$, the time complexity to do so is $\tilde{O}(T_3)$. 

For the segments $P_1, \ldots, P_i, Q_1, \ldots, Q_j$, we query the data structure in Lemma~\ref{lem:lem_multi_inc} with $S$ being the set of newly modified elements. The answer will be the most frequent element in the interval from $P_i$ to $Q_j$ that is not newly modified.
By Lemma~\ref{lem:lem_multi_inc} this takes 
$O(L) = O(N^{t2})$ time
per operation.\medskip

\noindent\textbf{Time and Space Complexity}\vspace{1pt}

In summary, the amortized cost per operation is 
\[\tilde{O}(N^{2-2t_1} + N^{t_2}+N^{t_3}+ N^{(1-t_3)(\frac{8}{5}+\frac{t_1}{1-t_3}+\frac{1}{5}\omega( \frac{t_1}{1-t_3})) - \frac{6}{5}t_2}).\]
To balance the terms, we set $t_1 = 1- \frac{1}{2} t_2$, and $t_3 = t_2$. The time complexity thus becomes 
\[\tilde{O}(N^{t_2}+ N^{(1-t_2)(\frac{8}{5}+\frac{1-0.5t_2}{1-t_2}+\frac{1}{5}\omega( \frac{1-0.5t_2}{1-t_2})) - \frac{6}{5}t_2}).\]
By observation, we can note that the optimum value of $\frac{1-0.5t_2}{1-t_2}$ lies in $[1.75, 2]$. Thus, we can plug in Corollary~\ref{cor:rect_mult} and use $t_2 = 0.655994$ to balance the two terms. This gives an $\tilde{O}(N^{0.655994})$ \textit{amortized} time per operation algorithm.

The space usage has two potential bottlenecks. The first is the space to store \sloppy $\mathcal{BST}_1, \ldots, \mathcal{BST}_{\frac{N}{T_1}}$ for handling infrequent elements, which is $\tilde{O}(\frac{N^2}{T_1})$.
The second is the space used 
by Lemma~\ref{lem:lem_multi_inc}, which is 
\small
\[
\tilde{O}(N^{-t_2/5} N^{(1-t_3)(18/5 + s - 4\omega(s)/5)}+N^{-3t_2/5} N^{(1-t_3) (9/5 + 2s - 2\omega(s)/5)}  + N^{-t_2/5} N^{(1-t_3)(8/5 + s + \omega(s) / 5})).
\]
\normalsize
By plugging in the values for $t_2, t_3$ and $s$, the space complexity becomes $\tilde{O}(N^{1.327997})$, with the $\tilde{O}(\frac{N^2}{T_1})$ term being the dominating term.
\medskip

\fussy

\noindent\textbf{Worst-Case Time Complexity}\vspace{1pt}

By applying the \textit{global rebuilding} of Overmars \cite{overmars1983design}, we can achieve a worst-case time bound. The basic idea is that after $T_2$ operations, we don't immediately rebuild the Min-Plus-Query-Witness data structure. Instead, we rebuild the data structure during the next $T_2$ operations, spreading the work evenly over each operation. To answer queries during these $T_2$ operations, we use the previous build of the Min-Plus-Query-Witness data structure.
By this technique, the per-operation runtime can be made worst-case.



\end{proof}

%% file: main.bbl
\begin{thebibliography}{10}

\bibitem{AlonGM97}
Noga Alon, Zvi Galil, and Oded Margalit.
\newblock On the exponent of the all pairs shortest path problem.
\newblock {\em J. Comput. Syst. Sci.}, 54(2):255--262, 1997.

\bibitem{BKMT05}
Prosenjit Bose, Evangelos Kranakis, Pat Morin, and Yihui Tang.
\newblock Approximate range mode and range median queries.
\newblock In {\em Annual Symposium on Theoretical Aspects of Computer Science},
  2005.

\bibitem{Bringmann17}
Karl Bringmann, Fabrizio Grandoni, Barna Saha, and Virginia Williams.
\newblock Truly sub-cubic algorithms for language edit distance and rna folding
  via fast bounded-difference min-plus product.
\newblock {\em SIAM Journal on Computing}, 48, 07 2017.
\newblock \href {http://dx.doi.org/10.1137/17M112720X}
  {\path{doi:10.1137/17M112720X}}.

\bibitem{BGJS11}
Gerth~St{\o}lting Brodal, Beat Gfeller, Allan J{\o}rgensen, and Peter Sanders.
\newblock Towards optimal range medians.
\newblock In {\em Theoretical Computer Science}, 2011.

\bibitem{Chan10}
Timothy~M. Chan.
\newblock More algorithms for all-pairs shortest paths in weighted graphs.
\newblock {\em SIAM J. Comput.}, 39(5):2075--2089, 2010.

\bibitem{C14}
Timothy~M. Chan, Stephane Durocher, Kasper~Green Larsen, Jason Morrison, and
  Bryan~T. Wilkinson.
\newblock Linear-space data structures for range mode query in arrays.
\newblock {\em Theory of Computing Systems}, 55:719--741, 2014.

\bibitem{chan17}
Timothy~M Chan and Konstantinos Tsakalidis.
\newblock Dynamic orthogonal range searching on the ram, revisited.
\newblock In {\em LIPIcs-Leibniz International Proceedings in Informatics},
  volume~77. Schloss Dagstuhl-Leibniz-Zentrum fuer Informatik, 2017.

\bibitem{Zein19}
Hicham El-Zein, Meng He, J.~Ian Munro, Yakov Nekrich, and Bryce Sandlund.
\newblock On approximate range mode and range selection.
\newblock In {\em 30th International Symposium on Algorithms and Computation
  (ISAAC 2019)}, 2019.

\bibitem{zhms2018}
Hicham El{-}Zein, Meng He, J.~Ian Munro, and Bryce Sandlund.
\newblock Improved time and space bounds for dynamic range mode.
\newblock In {\em Proceedings of the 26th Annual European Symposium on
  Algorithms}, pages 25:1--25:13, 2018.

\bibitem{elm11}
Amr Elmasry, Meng He, J~Ian Munro, and Patrick~K Nicholson.
\newblock Dynamic range majority data structures.
\newblock In {\em International Symposium on Algorithms and Computation}, pages
  150--159. Springer, 2011.

\bibitem{Fischer71}
Michael~J Fischer and Albert~R Meyer.
\newblock Boolean matrix multiplication and transitive closure.
\newblock In {\em n 12th Annual Symposium on Switching and Automata Theory},
  pages 129--131. IEEE, 1971.

\bibitem{gag17}
Travis Gagie, Meng He, and Gonzalo Navarro.
\newblock Compressed dynamic range majority data structures.
\newblock In {\em Data Compression Conference (DCC), 2017}, pages 260--269.
  IEEE, 2017.

\bibitem{legallmult}
Fran{\c{c}}ois~Le Gall.
\newblock Powers of tensors and fast matrix multiplication.
\newblock In {\em International Symposium on Symbolic and Algebraic
  Computation, {ISSAC} 2014, Kobe, Japan, July 23-25, 2014}, pages 296--303,
  2014.

\bibitem{gall2018improved}
Fran{\c{c}}ois~Le Gall and Florent Urrutia.
\newblock Improved rectangular matrix multiplication using powers of the
  coppersmith-winograd tensor.
\newblock In {\em Proceedings of the Twenty-Ninth Annual ACM-SIAM Symposium on
  Discrete Algorithms}, pages 1029--1046. SIAM, 2018.

\bibitem{GJLT10}
Mark Greve, Allan~Gr{\o}nlund J{\o}rgensen, Kasper~Dalgaard Larsen, and Jakob
  Truelsen.
\newblock Cell probe lower bounds and approximations for range mode.
\newblock In {\em International Colloquium on Automata, Languages, and
  Programming}, 2010.

\bibitem{HMN11}
Meng He, J.~Ian Munro, and Patrick~K. Nicholson.
\newblock Dynamic range selection in linear space.
\newblock In {\em International Symposium on Algorithms and Computation}, 2011.

\bibitem{Henzinger15}
Monika Henzinger, Sebastian Krinninger, Danupon Nanongkai, and Thatchaphol
  Saranurak.
\newblock Unifying and strengthening hardness for dynamic problems via the
  online matrix-vector multiplication conjecture.
\newblock In {\em Proceedings of the Forty-Seventh Annual ACM Symposium on
  Theory of Computing}, page 21–30, 2015.

\bibitem{Krizanc05}
Danny Krizanc, Pat Morin, and Michiel Smid.
\newblock Range mode and range median queries on lists and trees.
\newblock {\em Nordic Journal of Computing}, 2005.

\bibitem{le2012faster}
Fran{\c{c}}ois Le~Gall.
\newblock Faster algorithms for rectangular matrix multiplication.
\newblock In {\em 2012 IEEE 53rd annual symposium on foundations of computer
  science}, pages 514--523. IEEE, 2012.

\bibitem{lotti1983asymptotic}
Grazia Lotti and Francesco Romani.
\newblock On the asymptotic complexity of rectangular matrix multiplication.
\newblock {\em Theoretical Computer Science}, 23(2):171--185, 1983.

\bibitem{overmars1983design}
Mark~H Overmars.
\newblock {\em The design of dynamic data structures}, volume 156.
\newblock Springer Science \& Business Media, 1983.

\bibitem{mih10}
Mihai P{\u{a}}tra{\c{s}}cu and Emanuele Viola.
\newblock Cell-probe lower bounds for succinct partial sums.
\newblock In {\em Proceedings of the twenty-first annual ACM-SIAM symposium on
  Discrete Algorithms}, pages 117--122. Society for Industrial and Applied
  Mathematics, 2010.

\bibitem{Seidel95}
Raimund Seidel.
\newblock On the all-pairs-shortest-path problem in unweighted undirected
  graphs.
\newblock {\em J. Comput. Syst. Sci.}, 51(3):400--403, 1995.

\bibitem{Shoshan99}
Avi Shoshan and Uri Zwick.
\newblock All pairs shortest paths in undirected graphs with integer weights.
\newblock In {\em 40th Annual IEEE Symposium on Foundations of Computer
  Science}, pages 605--615, 1999.

\bibitem{vmmult}
Virginia {Vassilevska Williams}.
\newblock Multiplying matrices faster than {C}oppersmith-{W}inograd.
\newblock In {\em Proceedings of the 44th Symposium on Theory of Computing
  Conference, {STOC} 2012, New York, NY, USA, May 19 - 22, 2012}, pages
  887--898, 2012.

\bibitem{Williams20}
Virginia~Vassilevska Williams and Yinzhan Xu.
\newblock Truly subcubic min-plus product for less structured matrices, with
  applications.
\newblock In {\em Proceedings of the 2020 ACM-SIAM Symposium on Discrete
  Algorithms}, 2020.

\bibitem{Yuster09}
Raphael Yuster.
\newblock Efficient algorithms on sets of permutations, dominance, and
  real-weighted apsp.
\newblock In {\em 20th Annual ACM-SIAM Symposium on Discrete Algorithms}, pages
  950--957, 2009.

\bibitem{Zwick02}
Uri Zwick.
\newblock All pairs shortest paths using bridging sets and rectangular matrix
  multiplication.
\newblock {\em J. ACM}, 49(3):289--317, 2002.

\end{thebibliography}
